\documentclass[twocolumn,10pt]{IEEEtran}

\usepackage{hyperref}
\usepackage{amsmath,amssymb,amsthm}
\usepackage[pdftex]{graphicx}
\usepackage[caption=false,font=footnotesize]{subfig}
\usepackage{cite}
\usepackage{mathrsfs}
\usepackage{ifpdf}
\ifCLASSINFOpdf
  \usepackage[pdftex]{graphicx}
\else
  \usepackage[dvips]{graphicx}
\fi
\usepackage{wrapfig}
\usepackage{multirow}
\usepackage{algorithmic}
\usepackage[ruled, linesnumbered]{algorithm2e}
\usepackage{stfloats}
\usepackage{mdwmath}
\usepackage{mdwtab}
\usepackage{url}
\usepackage{array}
\usepackage{setspace}
\usepackage{color}
\usepackage{mdframed}
\usepackage{thmtools}

\newtheoremstyle{newstyle}      
{0pt} 
{0pt} 
{\mdseries} 
{} 
{\bfseries} 
{.} 
{3pt} 
{} 

\theoremstyle{newstyle}
\newtheorem{theorem}{Theorem}
\newtheorem{lemma}{Lemma}

\newtheorem{proposition}{Proposition}
\newtheorem{definition}{Definition}
\newtheorem{example}{Example}

\newcommand{\B}[1]{\mathbf{#1}}

\newcommand{\C}[1]{\mathcal{#1}}

\newcommand{\BB}[1]{\mathbb{#1}}

\newcommand{\myqed}{\hfill $\blacksquare$}
\newcommand{\head}{\mathrm{head}}
\newcommand{\tail}{\mathrm{tail}}
\newcommand{\mincut}{\mathrm{mincut}}

\newcommand{\inedge}{\mathrm{In}}
\newcommand{\outedge}{\mathrm{Out}}

\newcommand{\ie}{\textit{i.e.,~}}

\begin{document}

\title{On Routing-Optimal Networks for Multiple Unicasts}

\author{\IEEEauthorblockN{Chun Meng, Athina Markopoulou\\
EECS Department, University of California, Irvine \\
email: \{cmeng1,athina\}@uci.edu} \vspace{-1cm}
}

\maketitle

\begin{abstract}
In this paper, we consider the problem of multiple unicast sessions over a directed acyclic graph. It is well known that linear network coding is insufficient for achieving the capacity region, in the general case. However, there exist networks for which routing is sufficient to achieve the whole rate region, and we refer to them as {\em routing-optimal networks}. We identify a class of routing-optimal networks, which we refer to as {\em information-distributive networks}, defined by three topological features. Due to these features, for each rate vector
achieved by network coding, there is always a routing scheme such that it achieves the same rate vector, and the traffic transmitted through the network is exactly the information transmitted over the cut-sets between the sources and the sinks in the corresponding network coding scheme. We present examples of information-distributive networks, including some examples from (1) index coding and (2) from a single unicast session with hard deadline constraint.  
\end{abstract}

\section{Introduction \label{secIntro}}

In this paper, we consider network coding for multiple unicast sessions over directed acyclic graphs.
In general, non-linear network coding should be considered in order to achieve the whole rate region of network coding \cite{nonlinear_nc}.
Yet, there exist networks, for which routing is sufficient to achieve the whole rate region.
We refer to these networks as \textit{routing-optimal} networks.
We attempt to answer the following questions: 1) What are the distinct topological features of these networks? 2) Why do these features make a network routing-optimal?
The answers to these questions will not only explain which kind of networks can or cannot benefit from network coding, but will also deepen our understanding on how network topologies affect the rate region of network coding.

A major challenge is that there is currently no effective method to calculate the rate region of network coding.
Some researchers proposed to use information inequalities to approximate the rate region \cite{capacity_info_network}.
However, except for very simple networks, it is very difficult to use this approach since there is potentially an exponential number of inequalities that need to be considered.
\cite{cap_region_multisource} provides a formula to calculate the rate region by finding all possible entropy functions, which are vectors of an exponential number of dimensions, thus very difficult to solve even for simple networks.

In this paper, we employ a graph theoretical approach in conjunction with information inequalities to identify topological features of routing-optimal networks. 
Our high-level idea is as follows.
Consider a network code.
For each unicast session, we choose a cut-set $C$ between source and sink, and a set $\C{P}$ of paths from source to sink such that each path in $\C{P}$ passes through an edge in $C$.
Since the information transmitted from the source is totally contained in the information transmitted along the edges in $C$, we can think of  distributing the source information along the edges in $C$ (details will be explained later).
Moreover, we consider a routing scheme in which the traffic transmitted along each path $P\in \C{P}$ is exactly the source information distributed over the edge in $C$ that is traversed by $P$.
Such a routing scheme achieves the same rate vector as the network code.
However, since the edges might be shared among multiple unicast sessions, such a routing scheme might not satisfy the edge capacity constraints.
This suggests that the cut-sets and path-sets we choose for the unicast sessions should have special features.
These are essentially the features we are looking for to describe routing-optimal networks.

We make the following contributions:
\begin{itemize}
\item We identify a class of networks, called \textit{information-distributive} networks, which are defined by three topological features.
The first two features capture how the edges in the cut-sets are connected to the sources and the sinks, and the third feature captures how the paths in the path-sets overlap with each other.
Due to these features, given a network code, there is always a routing scheme such that it achieves the same rate vector as the network code, and the traffic transmitted through the network is exactly the source information distributed over the cut-sets between the sources and the sinks. 

\item We prove that if a network is information-distributive, it is routing-optimal. 
We also show that the converse is not true.
This indicates that the three features might be too restrictive in describing routing-optimal networks.

\item We present examples of information-distributive networks taken from the index coding problem \cite{index_code} and single unicast with hard deadline constraint.
\end{itemize}

We expect that our work will provide helpful insights towards characterizing all possible routing-optimal networks.

\section{Preliminaries \label{secPrelim}}

\subsection{Network Model}
The network is represented by an acyclic directed multi-graph $G=(V,E)$, where $V$ and $E$ are the set of nodes and the set of edges in the network respectively.
Edges are denoted by $e=(u,v,i) \in V \times V \times \BB{Z}_{\ge 0}$, or simply by $(u,v)$, where $v=\head(e)$ and $u=\tail(e)$.
Each edge represents an error-free and delay-free channel with capacity rate of one.
Let $\inedge(v)$ and $\outedge(v)$ denote the set of incoming edges and the set of outgoing edges at node $v$.

There are $K\ge 1$ unicast sessions in the network.
The $i$th unicast session is denoted by a tuple $\omega_i = (s_i,d_i)$, where $s_i$ and $d_i$ are the source and the sink of $\omega_i$ respectively.
The message sent from $s_i$ to $d_i$ is assumed to be a uniformly distributed random variable $Y_i$ with finite alphabet $\C{Y}_i = \{1,\cdots,\lceil 2^{nR_i} \rceil\}$, where $R_i$ is the source information rate at $s_i$.
All $Y_i$'s are mutually independent.
Given $1\le i \le j \le K$, denote $Y_{i:j} = \{Y_m: i\le m \le j\}$.
We assume $\inedge(s_i)=\outedge(d_i)=\emptyset$ for all $1\le i \le K$.

Let $\mincut(u,v,G)$ denote the minimum capacity of all cut-sets between two nodes $u$ and $v$.
Given two nodes $u,v$, let $\C{P}_{uv}$ denote the set of directed paths from $u$ to $v$.
The \textit{routing domain} of $\omega_i$, denoted by $G_i$, is the sub-graph induced by the edges of the paths in $\C{P}_{s_id_i}$.

\subsection{Routing Scheme}

A \textit{routing scheme} is a transmission scheme where each node only replicates and forwards the received messages onto its outgoing edges.
Define the following linear constraints:
\begin{flalign}
& \sum_{P \in \C{P}_{s_id_i}} f_i(P) \ge R'_i \hspace{2cm} \forall 1\le i \le K \label{eqRoutingCond1} \\
& \sum^K_{i=1} \sum_{P \in \C{P}_{s_id_i}, e\in P} f_i(P) \le 1 \hspace{1cm} \forall e \in E \label{eqRoutingCond2}
\end{flalign}
where $f_i(P)\in \BB{R}_{\ge 0}$ represents the amount of traffic routed through path $P$ for $\omega_i$.
A rate vector $\B{R}=(R'_i:1\le i \le K)\in \BB{R}^K_{\ge 0}$ is achievable by routing scheme if there exist $f_i(P)$'s such that (\ref{eqRoutingCond1}) and (\ref{eqRoutingCond2}) are satisfied.
The rate region of routing scheme, denoted by $\C{R}_r$, is the set of all rate vectors achievable by routing scheme.

\subsection{Network Coding Scheme}

A network coding scheme is defined as follows: \cite{cap_region_multisource}
\begin{definition}
\label{defNC}
An $(n,(\eta_e: e\in E),(R_i:1\le i\le K), (\delta_i: 1\le i \le K))$ \textit{network code} with block length $n$ is defined by:
\begin{enumerate}
\item for each $1\le i \le K$ and $e\in \outedge(s_i)$, a local encoding function: $\phi_e: \C{Y}_i \rightarrow \{1,\cdots,\eta_e\}$;

\item for each $v\in V-\{s_i,d_i:1\le i \le K\}$ and $e\in \outedge(v)$, a local encoding function: $\phi_e: \prod_{e'\in \inedge(v)} \{1,\cdots,\eta_{e'}\} \rightarrow \{1,\cdots,\eta_e\}$;

\item for each $1\le i \le K$, a decoding function: $\psi_i: \prod_{e' \in \inedge(d_i)} \{1,\cdots,\eta_{e'}\} \rightarrow \C{Y}_i$;

\item for each $1\le i \le K$, the decoding error for $\omega_i$ is $\delta_i = Pr(\tilde{\psi}_i(Y_{1:K}) \neq Y_i)$, where $\tilde{\psi}_i(Y_{1:K})$ is the value of $\psi_i$ as a function of $Y_{1:K}$.
\end{enumerate}
\end{definition}

Given $e\in E$, let $U_e = \tilde{\phi_e}(Y_{1:K})$, where $\tilde{\phi_e}(Y_{1:K})$ is the value of $\phi_e$ as a function of $Y_{1:K}$, denote the random variable transmitted along $e$ in a network code.
For a subset $C\subseteq E$, denote $U_C = \{U_e: e\in C\}$.

\begin{definition}
\label{defAchievable}
A rate vector $\B{R} = (R'_i: 1\le i \le K) \in \BB{R}^K_{\ge 0}$ is \textit{achievable} by network coding if for any $\epsilon > 0$, there exists for sufficiently large $n$, an $(n,(\eta_e: e\in E),(R_i:1\le i\le K), (\delta_i: 1\le i \le K))$ network code such that the following conditions are satisfied:
\begin{flalign}
& \frac{1}{n}\log\eta_e \le 1 + \epsilon \hspace{1.8cm} \forall e\in E \label{eqNCAchieve1} \\
& R_i \ge R'_i - \epsilon \hspace{2.4cm} \forall 1\le i \le K \label{eqNCAchieve2} \\
& \delta_i \le \epsilon \hspace{3.3cm} \forall 1\le i \le K \label{eqNCAchieve3}
\end{flalign}
The \textit{capacity region} achieved by network coding, denoted by $\C{R}_{nc}$, is the set of all rate vectors $\B{R}$ achievable by network coding.
\end{definition}

Given a network code that satisfies (\ref{eqNCAchieve1})-(\ref{eqNCAchieve3}), the following inequalities must hold:
\begin{flalign}
& \frac{1}{n} H(U_e) \le \frac{1}{n} \log(\eta_e) \le 1 + \epsilon \hspace{2cm} \forall e\in E \label{eqNCAchieveI1} \\
& \frac{1}{n} H(Y_i) = \frac{1}{n} \log(\lceil 2^{nR_i} \rceil) \ge R_i \ge R'_i - \epsilon \hspace{.3cm} \forall 1\le i \le K \label{eqNCAchieveI2} \\
& \frac{1}{n} I(Y_i;U_{\inedge(d_i)}) \ge (1-\epsilon)(R'_i-\epsilon) - \frac{1}{n} \hspace{.65cm} \forall 1\le i \le K \label{eqNCAchieveI3}
\end{flalign}
where (\ref{eqNCAchieveI3}) is due to Fano's Inequality:
\begin{flalign*}
& \frac{1}{n} I(Y_i;U_{\inedge(d_i)}) \ge \frac{1}{n}(H(Y_i) - \delta_i \log|\C{Y}_i| - 1) \\
=& \frac{1}{n} (1 - \delta_i) H(Y_i) - \frac{1}{n} \ge (1-\epsilon) (R'_i - \epsilon) - \frac{1}{n}
\end{flalign*}

\subsection{Routing-Optimal Networks}

Since routing scheme is a special case of network coding, $\C{R}_r \subseteq \C{R}_{nc}$.

\begin{definition}
A network is said to be \textit{routing-optimal}, if $\C{R}_{nc} = \C{R}_r$, \ie for such network, routing is sufficient to achieve the whole rate region of network coding.
\end{definition}

\section{A Class of Routing-Optimal Networks \label{secRoutingOpt}}

In this section, we present a class of routing-optimal networks, called \textit{information-distributive} networks.
We first use examples to illustrate the topological features of these networks, and show why they make the networks routing-optimal.
Then, we define these networks more rigorously.

\subsection{Illustrative Examples \label{subsecExample}}

\begin{example}
\label{ex1}
We start with the simplest case of single unicast.
It is well known that for this case, a network is always routing-optimal \cite{nc_first}.
In this example, we re-investigate this case from a new perspective in order to highlight some of the important features that make it routing optimal.
Let $m=\mincut(s_1,d_1,G)$, and $C=\{e_1,\cdots,e_m\}$ is a cut-set between $s_1$ and $d_1$.
Assume $R'_1\in \C{R}_{nc}$.
Therefore, for $\epsilon = \frac{1}{k} > 0$ ($k\in \BB{Z}_{> 0}$), there exists a network code such that (\ref{eqNCAchieve1})-(\ref{eqNCAchieve3}) are satisfied.
In the followings, all the random variables are defined in this network code.

One important feature of this network is that each path from $s_1$ to $d_1$ must pass through at least an edge in $C$.
Thus, $U_{\inedge(d_1)}$ is a function of $U_C$.
The following inequality holds:
\begin{flalign}
\begin{split}
\label{eqEx1Determine}
I(Y_1;U_{\inedge(d_1)}) \le I(Y_1;U_C)
\end{split}
\end{flalign}

The following equation holds:
\begin{flalign}
\label{eqEx1Distr}
I(Y_1;U_C) = \sum^m_{j=1} I(Y_1;U_{e_j} | U_{\{e_1,\cdots,e_{j-1}\}})
\end{flalign}
Intuitively, we can interpret (\ref{eqEx1Distr}) as follows:
$I(Y_1;U_{e_1})$ is the amount of information about $Y_1$ that can be obtained from $U_{e_1}$, $I(Y_1;U_{e_2}|U_{e_1})$ the amount of information about $Y_1$ that can be obtained from $U_{e_2}$, excluding those already obtained from $U_{e_1}$, and so on.
Hence, (\ref{eqEx1Distr}) can be seen as a ``distribution'' of the source information over the edges in $C$.
Moreover, for each $1\le j \le m$, we have:
\begin{flalign}
\label{eqEx1Cap}
I(Y_1; U_{e_j} | U_{\{e_1,\cdots,e_{j-1}\}}) \le H(U_{e_j})
\end{flalign}

Another important feature is that due to Menger's Theorem, there exist $m$ edge-disjoint paths, $P_1,\cdots,P_m$, from $s_1$ to $d_1$ such that $e_j \in P_j$ for $1\le j \le m$.
Due to this feature, we can construct a routing scheme by simply letting each $P_j$ transmit the information distributed on $e_j$:
\begin{flalign}
\label{eqEx1Routing}
f^{n,k}(P) = \begin{cases}
\frac{1}{n} I(Y_1;U_{e_j}|U_{\{e_1,\cdots,e_{j-1}\}}) & \text{if } P = P_j, 1\le j \le m \\
0 & \text{otherwise}.
\end{cases}
\end{flalign}
Clearly, due to (\ref{eqNCAchieveI1}) and (\ref{eqEx1Cap}), the above routing scheme satisfies the following inequalities:
\begin{flalign}
\label{eqEx1RoutingCap}
& f^{n,k}(P_j) \le \frac{1}{n} H(U_{e_j}) \le 1 + \frac{1}{k}
\end{flalign}
Moreover, due to (\ref{eqNCAchieveI3})-(\ref{eqEx1Distr}), we have:
\begin{flalign}
\begin{split}
\label{eqEx1RoutingRate}
& \sum_{P\in \C{P}_{s_1d_1}} f^{n,k}(P) = \sum^m_{j = 1} f^{n,k}(P_j) = \frac{1}{n} I(Y_1;U_C) \\
\ge & \frac{1}{n} I(Y_1;U_{\inedge(d_1)}) \ge \bigg(1-\frac{1}{k}\bigg) \bigg(R'_i - \frac{1}{k}\bigg) - \frac{1}{n}
\end{split}
\end{flalign}

Since $f^{n,k}(P_j)$ have an upper bound (see (\ref{eqEx1RoutingCap})), there exists a sub-sequence $(n_l,k_l)^{\infty}_{l=1}$ such that each sequence $(f^{n_l,k_l}(P_j))^{\infty}_{l=1}$ approaches a finite limit.
Define the following routing scheme:
\begin{flalign*}
f_1(P) = \begin{cases}
\lim_{l\rightarrow \infty} f^{n_l,k_l}(P) & \text{if } P = P_j (1\le j \le m); \\
0 & \text{otherwise}.
\end{cases}
\end{flalign*}
Due to (\ref{eqEx1RoutingCap}) and (\ref{eqEx1RoutingRate}), the above routing scheme satisfies (\ref{eqRoutingCond1}) and (\ref{eqRoutingCond2}).
Hence, $R'_1\in \C{R}_r$, which implies $\C{R}_{nc} \subseteq \C{R}_r$.
Therefore, the network is routing-optimal.
\myqed
\end{example}

\begin{figure}
\centering
\subfloat[2 unicast sessions \label{figEx0TwoUnicast}]{\includegraphics[scale=.45]{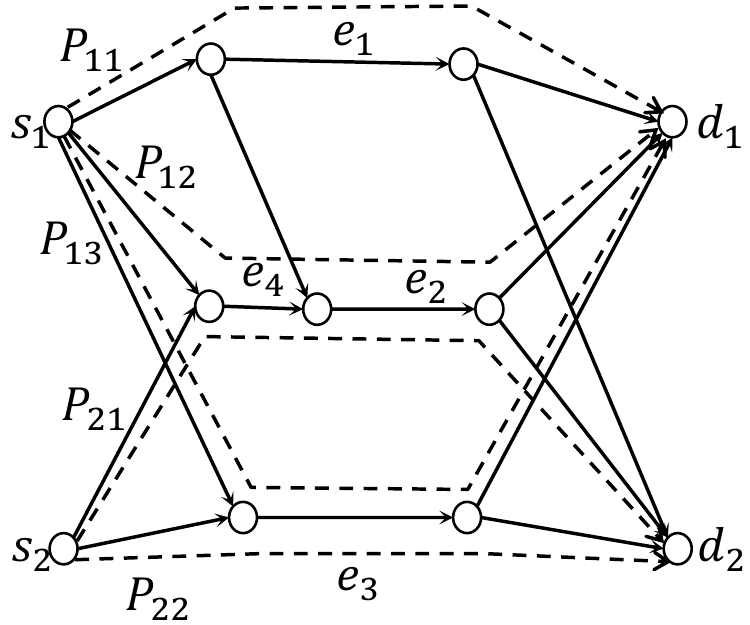}}
\subfloat[3 unicast sessions \label{figEx0ThreeUnicast}]{\includegraphics[scale=.45]{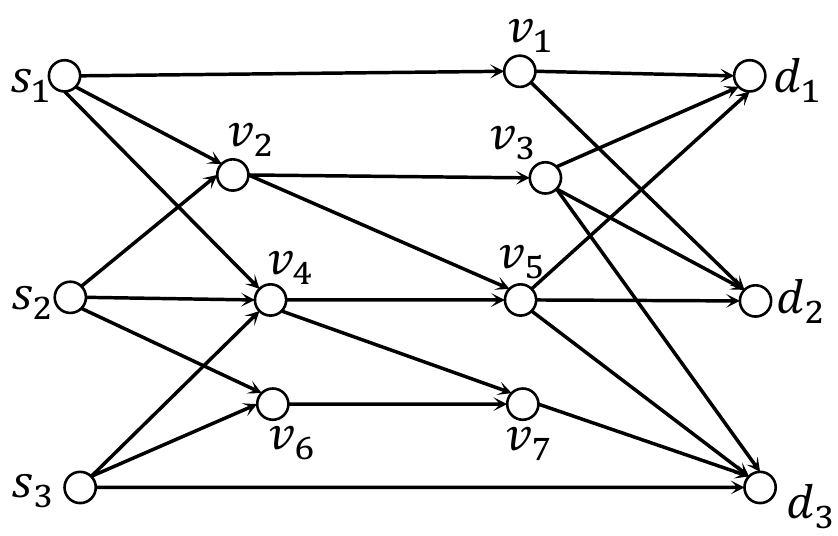}}
\caption{Examples of information-distributive networks, where $s_i,d_i$ ($1\le i \le 3$) are the source and the sink of the $i$th unicast session respectively. \label{figRoutingOptExample}}
\end{figure}

As shown above, two features are essential in making a network with single-unicast routing-optimal.
The first feature is the existence of a cut-set such that each path from the source to the sink must pass through an edge in the cut-set.
Due to this feature, the source information contained in $U_{\inedge(d_1)}$ can be completely obtained from the messages transmitted through the cut-set $C$ (see (\ref{eqEx1Determine})).
The second feature is the existence of edge-disjoint paths $P_1,\cdots,P_m$, each of which passes through exactly one edge in $C$.
Due to this feature, a routing scheme can be constructed such that the traffic transmitted along the  paths $P_1,\cdots,P_m$ is exactly the information distributed on the edges in $C$ (see (\ref{eqEx1Routing})).
These two features together guarantee that the routing scheme achieves the same rate as network coding (see (\ref{eqEx1RoutingCap}), (\ref{eqEx1RoutingRate})).

However, extending these features to multiple unicast sessions is not straightforward.
One difference from single unicast is that $U_{\inedge(d_i)}$ may not be a function of $U_C$, where $C$ is a cut-set between $s_i$ and $d_i$, and thus (\ref{eqEx1Determine}) might not hold.
Another difference is that the information from multiple unicast sessions might be distributed on an edge, and thus (\ref{eqEx1Cap}) might not hold.
Moreover, the paths for multiple unicast sesssions might overlap with each other, and thus (\ref{eqEx1RoutingCap}) might not hold.
These differences suggest that the cut-sets and the paths, over which a routing scheme is to be constructed, should have additional features in order for the resulting routing scheme to achieve the same rate vector as network coding.
We use an example to illustrate some of these features.

\begin{example}
\label{ex2}
Consider the network shown in Fig. \ref{figEx0TwoUnicast}a.
Consider an arbitrary rate vector $\B{R}=(R'_1,R'_2) \in \C{R}_{nc}$.
Therefore, for $\epsilon=\frac{1}{k}$ ($k\in \BB{Z}_{> 0}$), there exists a network code that satisfies (\ref{eqNCAchieve1})-(\ref{eqNCAchieve3}).
In the sequel, all the random variables are defined in this network code.

For $\omega_1$, we choose a cut-set $C_1=\{e_1,e_2,e_3\}$ between $s_1$ and $d_1$, and a set of paths $\C{P}_1= \{P_{11},P_{12},P_{13}\}$ that pass through $e_1,e_2,e_3$ respectively; 
for $\omega_2$, we choose a cut-set $C_2=\{e_2,e_3\}$ between $s_2$ and $d_2$, and a set of paths $\C{P}_2=\{P_{21},P_{22}\}$ that pass through $e_2,e_3$ respectively.

We first investigate $C_1,C_2$.
One important feature is that each path from $s_2$ to $d_1$ passes through at least an edge in $C_1$.
Thus, $C_1$ is also a cut-set between $\{s_1,s_2\}$ and $d_1$, and $U_{\inedge(d_1)}$ is a function of $U_{C_1}$.
Hence, we have:
\begin{flalign}
\label{eqEx2Determine1}
I(Y_1;U_{\inedge(d_1)}) \le I(Y_1;U_{C_1})
\end{flalign}
Moreover, $\outedge(s_1) \cup C_2$ is a cut-set between $\{s_1,s_2\}$ and $d_2$, and $U_{\outedge(s_1)}$ is a function of $Y_1$.
Hence $U_{\inedge(d_2)}$ is a function of $Y_1,U_{C_2}$, which implies:
\begin{flalign}
\label{eqEx2Determine2}
I(Y_2;U_{\inedge(d_2)} | Y_1) \le I(Y_2;U_{C_2} | Y_1) 
\end{flalign}
We distribute the source information over $C_1,C_2$ as follows:
\begin{flalign}
\label{eqEx2Distr}
\begin{split}
& I(Y_1;U_{C_1}) = I(Y_1;U_{e_1}) + I(Y_1;U_{e_2} | U_{e_1}) \\
& \hspace{2cm} + I(Y_1;U_{e_3}|U_{\{e_1,e_2\}}) \\
& I(Y_2;U_{C_2} | Y_1) = I(Y_2;U_{e_2} | Y_1) + I(Y_2;U_{e_3} | Y_1, U_{e_2})
\end{split}
\end{flalign}

Another feature about $C_1,C_2$ is that edge $e_1$ is connected to only one source $s_1$, and thus $U_{e_1}$ is a function of $Y_1$.
As shown below, this feature guarantees that the information distributed on an edge $e\in C_1\cup C_2$ is completely contained in $U_e$.
First, for $e_1$, it can be easily seen that:
\begin{flalign}
\label{eqEx2Cap1}
I(Y_1;U_{e_1}) \le H(U_{e_1})
\end{flalign}
For $e_2$, we have:
\begin{flalign}
\label{eqEx2Cap2}
\begin{split}
& I(Y_1; U_{e_2} | U_{e_1}) + I(Y_2; U_{e_2} | Y_1) \\
\overset{(b)}{=}& I(Y_1; U_{e_2} | U_{e_1}) + I(Y_2; U_{e_2} | Y_1, U_{e_1}) \\
=& I(Y_1,Y_2; U_{e_2} | U_{e_1}) \le H(U_{e_2})
\end{split}
\end{flalign}
where $(b)$ is due to the fact that $U_{e_1}$ is a function of $Y_1$, and thus, $I(Y_2; U_{e_2} | Y_1) = I(Y_2; U_{e_2} | Y_1, U_{e_1})$.
Similarly, for $e_3$, we have:
\begin{flalign}
\label{eqEx2Cap3}
\begin{split}
& I(Y_1; U_{e_3} | U_{\{e_1,e_2\}}) + I(Y_2; U_{e_3} | Y_1, U_{e_2}) \\
\overset{(c)}{=}& I(Y_1; U_{e_3} | U_{\{e_1,e_2\}}) + I(Y_2; U_{e_3} | Y_1, U_{\{e_1,e_2\}}) \\
=& I(Y_1,Y_2; U_{e_3} | U_{\{e_1,e_2\}}) \le H(U_{e_3})
\end{split}
\end{flalign}
where $(c)$ is again due to the fact that $U_{e_1}$ is a function of $Y_1$.

Next, we investigate $\C{P}_1, \C{P}_2$.
One important feature is that if $P\in \C{P}_1$ overlaps with $P'\in \C{P}_2$, $P\cap C_1 = P'\cap C_2$.
For example, $P_{12}$ overlaps with $P_{21}$, and $P_{12} \cap C_1 = P_{21} \cap C_2 = \{e_2\}$.
This feature ensures that the information distributed over $C_1,C_2$ can be further distributed over the paths in $\C{P}_1,\C{P}_2$.
To see this, we construct the following routing scheme:
\begin{flalign*}
& f^{n,k}_1(P) = \begin{cases}
\frac{1}{n} I(Y_1;U_{e_j}| U_{\{e_1,\cdots,e_{j-1}\}}) & \text{if } P=P_{1j},1\le j \le 3 \\
0 & \text{otherwise}.
\end{cases} \\
& f^{n,k}_2(P) = \begin{cases}
\frac{1}{n} I(Y_2;U_{e_2} | Y_1) & \text{if } P=P_{21}; \\
\frac{1}{n} I(Y_2;U_{e_3} | Y_1, U_{e_2}) & \text{if } P=P_{22}; \\
0 & \text{otherwise}.
\end{cases} 
\end{flalign*}
Due to (\ref{eqEx2Cap1})-(\ref{eqEx2Cap3}), we can derive that for each $e\in C_1\cup C_2$,
\begin{flalign}
\label{eqEx2RoutingCap}
\sum^2_{i=1} \sum_{P\in \C{P}_{s_id_i}, e\in P} f^{n,k}_i(P) \le \frac{1}{n}H(U_e) \le 1 + \frac{1}{k}
\end{flalign}
For $e_4$, we have:
\begin{flalign*}
&\sum^2_{i=1} \sum_{P\in \C{P}_{s_id_i}, e_4\in P} f^{n,k}_i(P) \\
=& f^{n,k}_1(P_{12}) + f^{n,k}_2(P_{21}) \le \frac{1}{n}H(U_{e_2}) \le 1 + \frac{1}{k}
\end{flalign*}
Likewise, we can prove that (\ref{eqEx2RoutingCap}) holds for all the other edges of the paths in $\C{P}_1 \cup \C{P}_2$.
Due to (\ref{eqEx2Determine1})-(\ref{eqEx2Distr}), the following inequalities hold for $i=1,2$
\begin{flalign}
\label{eqEx2RoutingRate}
\sum_{P\in \C{P}_{s_id_i}} f^{n,k}_i(P) \ge \bigg(1 - \frac{1}{k}\bigg)\bigg(R'_i-\frac{1}{n}\bigg) + \frac{1}{n}
\end{flalign}
By (\ref{eqEx2RoutingCap}), there exists a sub-sequence $(n_l,k_l)^{\infty}_{l=1}$ such that for all $P\in \C{P}_1\cup \C{P}_2$ and $i=1,2$, the sub-sequence $(f^{n_l,k_l}_i(P))^{\infty}_{l=1}$ approaches a finite limit.
Define a routing scheme:
\begin{flalign}
f_i(P) = \begin{cases}
\lim_{l\rightarrow \infty} f^{n_l,k_l}_i(P) & \text{if } P \in \C{P}_i, i=1,2; \\
0 & \text{otherwise.}
\end{cases}
\end{flalign}
Due to (\ref{eqEx2RoutingCap}) and (\ref{eqEx2RoutingRate}), $f_i(P)$ satisfies (\ref{eqRoutingCond1}) and (\ref{eqRoutingCond2}).
Hence, $\B{R} \in \C{R}_r$, and $\C{R}_{nc} \subseteq \C{R}_r$.
The network is routing-optimal.
\myqed
\end{example}

\subsection{Information Distributive Networks \label{subsecInfoDistr}}

In this subsection, we present the definition of information-distributive networks.
Similarly to single unicast, for each unicast session $\omega_i$ ($1\le i \le K$), we choose a cut-set $C_i$ between $s_i$ and $d_i$ such that $|C_i|=\mincut(s_i,d_i,G_i)$, and a set of paths $\C{P}_i$ from $s_i$ to $d_i$.
The collection of these cut-sets, denoted by $\C{W} = (C_i)^K_{i=1}$, is called a \textit{cut-set sequence}, and the collection of these path-sets, denoted by $\C{K}=(\C{P}_i)^K_{i=1}$, is called a \textit{path-set sequence}.
For instance, in Example \ref{ex2}, we choose a cut-set sequence $\C{W}=(C_i)^2_{i=1}$, where $C_1=\{e_1,e_2,e_3\}$ is a cut-set between $s_1$ and $d_1$, and $C_2=\{e_2,e_3\}$ is a cut-set between $s_2$ and $d_2$, and a path-set sequence $\C{K}=(\C{P}_i)^2_{i=1}$, where $\C{P}_1$ is a path-set from $s_1$ to $d_1$, and $\C{P}_2$ a path-set from $s_2$ to $d_2$.
Moreover, we arrange the edges in each cut-set in $\C{W}$ in some ordering.
For instance, in Example \ref{ex2}, we arrange the edges in $C_1$ in the ordering $T_1=(e_1,e_2,e_3)$, and the edges in $C_2$ in the ordering $T_2=(e_2,e_3)$.
Each such ordering is called a permutation of the edges in the corresponding cut-set.
The collection of these permutations, denoted $\C{T} = (T_i)^K_{i=1}$, is called a \textit{permutation sequence}.
For $e\in C_i$, let $T_i(e)$ denote the subset of edges before $e$ in $T_i$.
For $e\in E$, define $\C{W}(e)=\{C_i \in \C{W}: e\in C_i\}$, and $\alpha(e)$ the largest index of the source to which $\tail(e)$ is connected.
The first feature is described below.

Next, we formalize the three features we have shown in Example \ref{ex2}.
The first feature is described below.

\begin{definition}
\label{defCumulative}
Given a cut-set sequence $\C{W}$, if for all $1\le i < j \le K$, each path from $s_j$ to $d_i$ must pass through an edge in $C_i$, we say that $\C{W}$ is \textit{cumulative}.
\end{definition}

This feature guarantees that the source information contained in the incoming messages at each sink $d_i$ can be completely obtained from $Y_{1:i-1},U_{C_i}$.
\begin{lemma}
\label{lemmaDetermine}
Consider a network code as defined in Definition \ref{defNC}.
If $\C{W}$ is a cumulative cut-set sequence, then for each $1\le i \le K$, $Y_i$ is a function of $Y_{1:i-1}, U_{C_i}$, and the following inequality holds:
\begin{flalign}
\label{eqDetermine}
I(Y_i; U_{\inedge(d_i)} | Y_{1:i-1}) \le I(Y_i; U_{C_i} | Y_{1:i-1})
\end{flalign}
\end{lemma}
\begin{proof}
See Appendix \ref{appInfoDistr}.
\end{proof}

Given a cumulative cut-set sequence $\C{W}$ and a permutation sequence $\C{T}$ for $\C{W}$, we can distribute the source information $Y_i$ over the edges in $C_i$ as follows:
\begin{flalign}
\label{eqInfoDistr}
I(Y_i; U_{C_i} | Y_{1:i-1}) = \sum_{e\in C_i} I(Y_i; U_e | Y_{1:i-1}, U_{T_i(e)})
\end{flalign}

The second feature is presented below.
Without loss of generality, let $\C{W}(e) = \{C_{n_1},\cdots,C_{n_k}\}$, where $1\le n_1 < \cdots < n_k \le K$.

\begin{definition}
\label{defDistributive}
Given a cut-set sequence $\C{W}$, we say that it is \textit{distributive} if there exists a permutation sequence $\C{T}$ for $\C{W}$ such that for each $e\in \bigcup^K_{i=1}C_i$, the following conditions are satisfied: for all $1\le j < k$,
\begin{flalign}
& \alpha(e') \le n_k \hspace{1.9cm} \forall e'\in T_{n_{j+1}}(e) - T_{n_j}(e) \label{eqDistr1} \\ 
& \alpha(e') \le n_{j+1} - 1 \hspace{1cm} \forall e'\in T_{n_j}(e) - T_{n_{j+1}}(e) \label{eqDistr2}
\end{flalign}
\end{definition}

As shown in Example \ref{ex2}, let $T_1=(e_1,e_2,e_3)$, and $T_2=(e_2,e_3)$.
For $e_3$, $\C{W}(e_3)=\{C_1,C_2\}$, $T_2(e_3)-T_1(e_3) = \emptyset$, and thus, (\ref{eqDistr1}) is trivially satisfied;  $T_1(e_3)-T_2(e_3) = \{e_1\}$, $\alpha(e_1)=1$, and (\ref{eqDistr2}) is satisfied.
Similarly, we can verify other edges.
Hence, $\C{W}$ is distributive.

The above two features ensure that the information from multiple unicast sessions that is distributed on an edge $e\in \bigcup^K_{i=1} C_i$ can be completely obtained from $U_e$.

\begin{lemma}
\label{lemmaInfoDistr}
Consider a network code as defined in Definition \ref{defNC}.
Given a cumulative cut-set sequence $\C{W}$, if $\C{W}$ is distributive, for each $e\in \bigcup^K_{i=1} C_i$, the following inequality holds:
\begin{flalign}
\label{eqTotalDistr}
\sum_{1\le i\le K, e\in C_i} I(Y_i; U_e | Y_{1:i-1}, U_{T_i(e)}) \le H(U_e)
\end{flalign}
\end{lemma}
\begin{proof}
See Appendix \ref{appInfoDistr}.
\end{proof}

The third feature is presented below.

\begin{definition}
\label{defExtendable}
Given a path-set sequence $\C{K}$ for $\C{W}$, we say that $\C{K}$ is \textit{extendable}, if for all $1\le i < j \le K$, $P_1\in \C{P}_i$ and $P_2\in \C{P}_j$ such that $P_1$ overlaps with $P_2$, $P_1\cap C_i = P_2 \cap C_j$.
\end{definition}

As shown in Example \ref{ex2}, let $\C{K}=\{\C{P}_1,\C{P}_2\}$.
Clearly, we have $P_{12} \cap P_{21} = \{e_2,e_4\}$, $P_{13} \cap C_1 = P_{21} \cap C_2 = \{e_2\}$, and $P_{13} \cap P_{22} = \{e_3\}$, $P_{13} \cap C_1 = P_{22} \cap C_2 = \{e_3\}$.
Thus, $\C{K}$ is extendable.

\begin{definition}
\label{defInfoDistributive}
A network with multiple unicast sessions is said to be \textit{information-distributive}, if there exist a cumulative and distributive cut-set sequence $\C{W}$, and an extendable path-set sequence $\C{K}$ for $\C{W}$ in the network. 
\end{definition}

As shown in the next theorem, the three features together guarantee that the network is routing-optimal.

\begin{theorem}
\label{thInfodistr}
If a network is information-distributive, it is routing-optimal.
\end{theorem}
\begin{proof}
See Appendix \ref{appInfoDistr}.
\end{proof}

\begin{example}
Consider the network shown in Fig. \ref{figEx0ThreeUnicast}.
Define the following cut-sets:
\begin{flalign*}
& C_1=\{(s_,v_1), (v_2,v_3), (v_4,v_5)\} \\
& C_2=\{(v_2,v_3), (v_4,v_5)\} \\
& C_3=\{(v_6,v_7), (s_3,d_3)\}
\end{flalign*}
Define $\C{W}=(C_i)^3_{i=1}$.
Define the following paths: 
\begin{flalign*}
& P_{11}=\{(s_1,v_1), (v_1,d_1)\} \\
& P_{12}=\{(s_1,v_2), (v_2,v_3), (v_3,d_1)\} \\
& P_{13}=\{(s_1,v_4), (v_4,v_5), (v_5,d_1)\} \\
& P_{21}=\{(s_2,v_2), (v_2,v_3), (v_3,d_2)\} \\
& P_{31}=\{(s_3,v_6), (v_6,v_7), (v_7,d_3)\} \\
& P_{33}=\{(s_3,d_3)\}
\end{flalign*}
Define $\C{K}=\{\{P_{11},P_{12},P_{13}\}$,$\{P_{21},P_{22}\}$,$\{P_{31},P_{32}\}\}$.
It can be verified that $\C{W}$ is cumulative and distributive, and $\C{K}$ is extendable.
The network is information-distributive. \myqed
\end{example}

\section{More Examples \label{secExample}}

\subsection{Index Coding}

We consider a multiple-unicast version of index coding problem \cite{index_code}.
In this problem, there are $K$ terminals $t_1,\cdots,t_K$, a broadcast station $s$, and $K$ source messages $X_1,\cdots,X_K$, all available at $s$.
All $X_i$'s are mutually independent random variables uniformly distributed over alphabet $\C{X}_i=\{1,\cdots,2^m\}$.
Each terminal requires $X_i$, and has acquired a subset of source messages $\C{H}_i$ such that $X_i\notin \C{H}_i$.
$s$ uses an encoding function $\phi: \prod^K_{i=1} \C{X}_i \rightarrow \{1,\cdots,2^l\}$ to encode the source messages, and broadcasts the encoded message to the terminals through an error-free broadcast channel.
Each $t_i$ uses a decoding function $\psi_i$ to decode $X_i$ by using the received message and the messages in $\C{H}_i$.
The encoding function $\phi$ and the decoding functions $\psi_i$'s are collectively called an index code, and $l$ is the length of this index code.
The minimum length of an index code is denoted by $l_{min}$.

This index coding problem can be cast to a multiple-unicast network coding problem over a network $G_1=(V_1,E_1)$, where $V_1=\{s_i,d_i: 1\le i \le K\} \cup \{u,v\}$, $E_1 = \{(s_i,u),(v,d_i): 1\le i\le K\} \cup \{(u,v)\} \cup \{(s_j,d_i): X_j \in \C{H}_i\}$.
The $K$ unicast sessions are $(s_1,d_1),\cdots,(s_K,d_K)$.
It can be verified that there exists an index code of length $l$, if and only if $\B{R}=(\frac{l}{m}, \cdots, \frac{l}{m})$ is achievable by network coding in $G_1$.

Let $C_i=\{(u,v)\}$, $P_i=\{(s_i,u),(u,v),(v,d_i)\}$.
Define $\C{W}=(C_i)^K_{i=1}$ and $\C{K}=(\C{P}_i)^K_{i=1}$, where $\C{P}_i=\{P_i\}$.
Since each $C_i$ contains only one edge, $\C{W}$ is distributive.
Meanwhile, since all $P_i$'s overlap at $(u,v)$, $\C{K}$ is extendable.

The following theorem states that if the optimal solution to the index coding problem is to let the broadcast station transmit raw packet, \ie no coding is needed, then the corresponding multiple-unicast network is information-distributive, and the converse is also true.

\begin{theorem}
\label{thIndexCode}
$l_{min} = mK$ if and only if $\C{W}$ is cumulative, \ie $G_1$ is information-distributive.
\end{theorem}
\begin{proof}
See Appendix \ref{appProofEx}.
\end{proof}

\begin{figure}
\centering
\includegraphics[scale=.4]{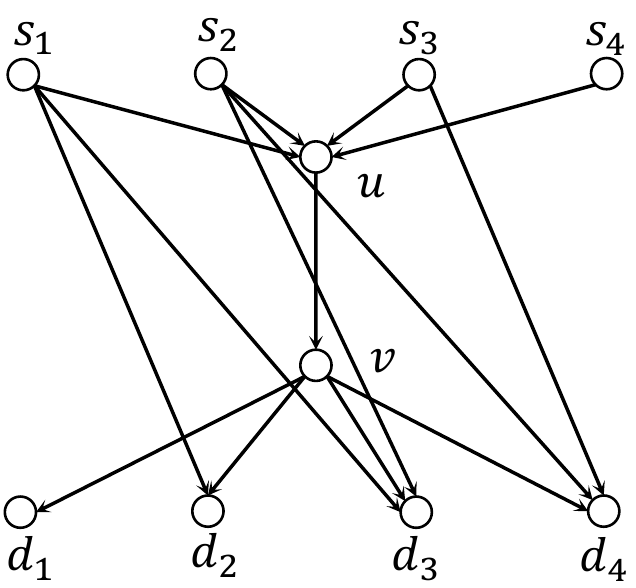}
\caption{The equivalent network coding problem for an index coding problem. The network is information-distributive, and thus no coding is needed in the index coding problem. \label{figIndexCode}}
\end{figure}

\begin{example}
\label{exIndexCode}
In Fig. \ref{figIndexCode}, we show an example of $G_1$, which corresponds to an index coding problem defined by: $\C{H}_1=\emptyset$, $\C{H}_2=\{X_1\}$, $\C{H}_3=\{X_1,X_2\}$, and $\C{H}_4=\{X_2,X_3\}$.
Clearly, $\C{W}$ is cumulative, and thus $l_{min}=mK$. \myqed
\end{example}

\subsection{Single Unicast with Hard Deadline Constraint}

In this example, we consider the network coding problem for a single-unicast session $(s,d)$ over a network $G=(V,E)$, where each edge $e$ is associated with a delay $d_e \in \BB{Z}_{>0}$, and each node has a memory to hold received data.
Given a directed path $P$, let $d(P)=\sum_{e\in P} d_e$ denote its delay.
For $e\in E$, let $\delta(e)$ denote the minimum delay of directed paths from $s$ to $\tail(e)$.
The data transmission in the network proceeds in time slots.
The messages transmitted from $s$ is represented by a sequence $(Y[t])^K_{t=0}$, where $Y[t]$ is a uniformly distributed random variable, and represents the message transmitted from $s$ at time slot $t$.
All $Y[t]$'s are mutually independent.
We require that each $Y[t]$ must be received by $d$ within $\tau$ time slots. 
Otherwise, it is regarded as useless, and is discarded.
This problem was first proposed by \cite{kodialam2002allocating}\cite{minghua_unpublished}.
Recently, it has been shown that network coding can improve throughput by utilizing over-delayed information \cite{delay_constrained}.

This problem can be cast to an equivalent network coding problem for multiple unicast sessions.
We construct a time-extended graph $\tilde{G} = (\tilde{V}, \tilde{E})$ as follows:
the node set is $\tilde{V}=\{s_t,d_t:0\le t \le K\} \cup \{v[t]: 0 \le t \le K+\tau\}$;
for each $e=(u,v)\in E$ and $0\le t \le K+\tau-d_e$, we add an edge $e[t]=(u[t],v[t+d_e])$ to $\tilde{E}$;
for $u\in V$, and $0\le t \le K+\tau-1$, we add $M$ edges from $u[t]$ to $u[t+1]$, where $M$ is the amount of memory available at $u$;
for each $0\le t \le K$, we add $J$ edges from $s_t$ to $s[t]$ and $J$ edges from $d[t+\tau]$ to $d_t$, where $J$ is a sufficiently large integer.
Thus, the original single unicast session $(s,d)$ is cast to $K+1$ unicast sessions $(s_0,d_0),\cdots,(s_K,d_K)$ over $\tilde{G}$.

Let $\tilde{G}[t]$ denote the routing domain for $(s_t,d_t)$, and $m=\mincut(s_0,d_0,\tilde{G}[0])$.
It can be seen that each $\tilde{G}[t]$ is simply a time-shifted version of $\tilde{G}[0]$.
Given a subset of edges $U\subseteq \tilde{E}$, define $U[t]=\{(u[k+t],v[l+t]): (u[k],v[l])\in U\}$.
Let $C=\{e_j[t_j]:1\le j\le m\}$ be a cut-set between $s_0$ and $d_0$ such that $e_j\in E$ for $1\le j \le m$, and  $\C{P}=\{P_j,\cdots,P_m\}$ a set of edge disjoint paths from $s_0$ to $d_0$ such that $e_j[t_j]\in P_j$ for $1\le j \le m$.
Let $\C{P}[t]=\{P[t]: P\in \C{P}\}$.
We consider the cut-set sequence $\C{W}=(C[t])^K_{t=0}$, and the path-set sequence $\C{K}=(\C{P}[t])^K_{t=0}$.

\begin{lemma}
\label{lemmaDelayCumulative}
$\C{W}$ is cumulative.
\end{lemma}
\begin{proof}
See Appendix \ref{appProofEx}.
\end{proof}

Given $U\subseteq \tilde{E}$, a \textit{recurrent} sequence of $U$ is a sequence consisting of all the edges in $U$ that are time-shifted versions of the same edge.
$C[0]$ is said to be \textit{distributive} if there is a re-indexing of the edges in $C[0]$ such that for each recurrent sequence $(e_p[t_{n_j}])^k_{j=1}$ of $C[0]$, the following conditions are satisfied:
\begin{enumerate}
\item for each $1< j\le k$, if $e_q[t_q]\in C[0]$ lies before $e_p[t_{n_j}]$, and $e_q[t_q-t_{n_j}+t_{n_{j-1}}] \notin C[0]$, then $t_q - \delta(e_q) \le t_{n_j} - t_{n_{j-1}} - 1$;

\item for each $1\le j< k$, if $e_q[t_q]\in C[0]$ lies before $e_p[t_{n_j}]$, and $e_q[t_q+t_{n_{j+1}}-t_{n_j}] \notin C[0]$, then $t_q - \delta(e_q) \le t_{n_j} - t_{n_1}$.
\end{enumerate}
$\C{P}$ is said to be \textit{extendable} if for all $P_i,P_j\in \C{P}$ and $e[k], e[l]\in \tilde{E}$ such that $e[k]\in P_i$ and $e[l]\in P_j$, $e_i=e_j$ and $t_i-t_j=k-l$.

\begin{theorem}
\label{thDelayInfoDistributive}
If $C[0]$ is distributive, and $\C{P}$ is extendable, $\tilde{G}$ is information-distributive, and thus is routing-optimal.
\end{theorem}
\begin{proof}
See Appendix \ref{appProofEx}.
\end{proof}

\begin{figure}
\centering
\subfloat[Original network\label{figEx3Net}]{\includegraphics[scale=.4]{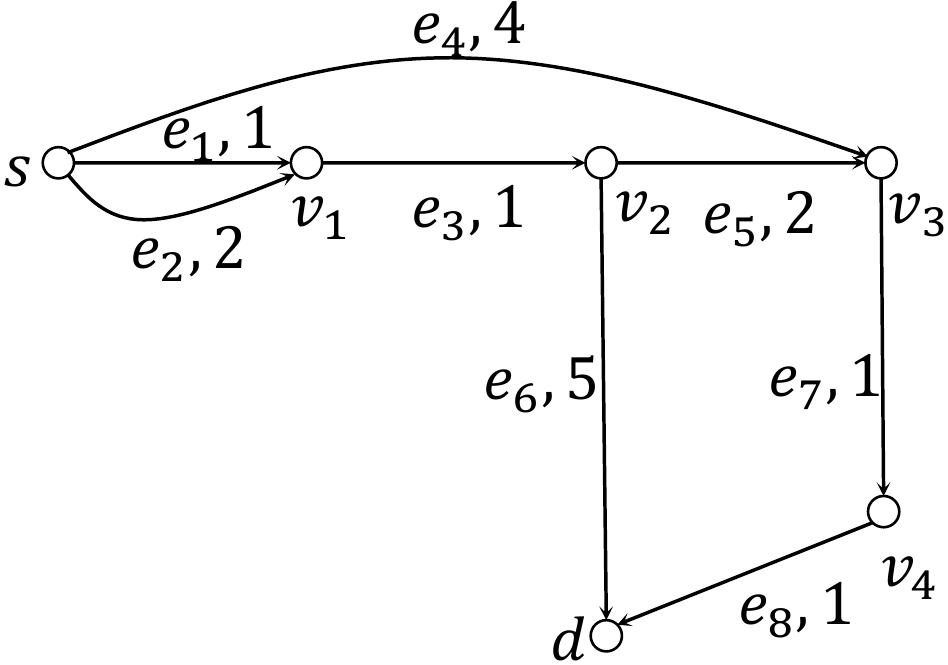}} 
\subfloat[Routing domain for $(s_0,d_0)$\label{figEx3RoutingDomain}]{\includegraphics[scale=.4]{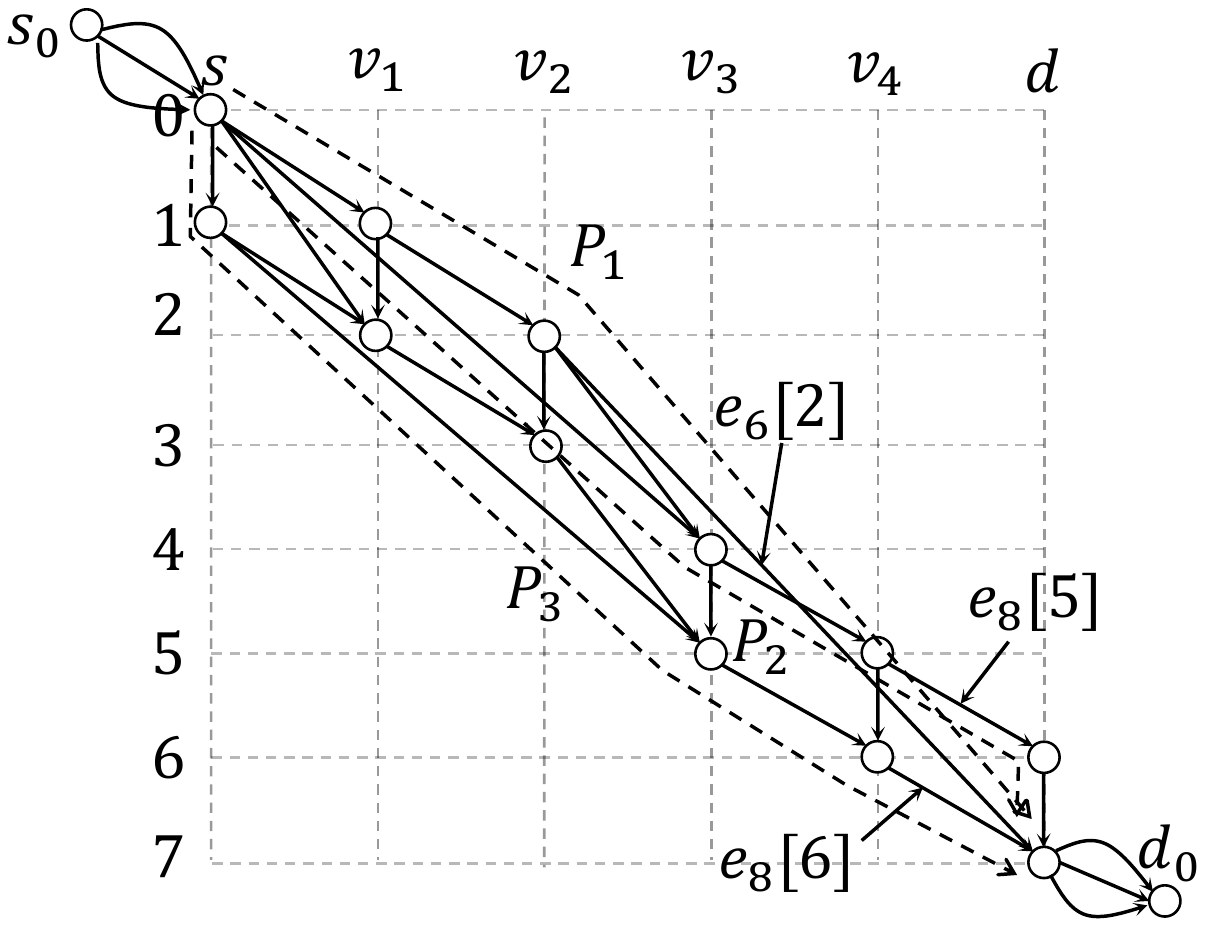}}
\caption{An example of single unicast with deadline constraint $\tau=7$.
(a) shows an network with a single unicast $(s,d)$, where $e_k,i$ denotes the alias of an edge and its corresponding delay respectively.
(b) shows the routing-domain between $s_0$ and $d_0$ over the corresponding time-extended graph $\tilde{G}$, where the node at coordinate $(v,t)$ is $v[t]$.
In this routing-domain, $C[0]=\{e_8[5],e_6[2],e_8[6]\}$ is distributive, and $\C{P}=\{P_1,P_2,P_3\}$ is extendable.
Hence, $\tilde{G}$ is information-distributive, and therefore, routing-optimal.}
\end{figure}

\begin{example}
\label{exNCDelay}
In Fig. \ref{figEx3Net}, we show an example of single unicast with delay constraint $\tau=7$.
In Fig. \ref{figEx3RoutingDomain}, we show the routing domain $\tilde{G}[0]$ for $(s_0,d_0)$.
Let $C[0]=\{e_8[5],e_6[2],e_8[6]\}$, and $\C{P}=\{P_1,P_2,P_3\}$, where $P_1,P_2,P_3$ are marked as black dashed lines in Fig. \ref{figEx3RoutingDomain}.
It can be verified that $C[0]$ is distributive, and $\C{P}$ is extendable.
Thus, according to Theorem \ref{thDelayInfoDistributive}, $\tilde{G}$ is information-distributive. \myqed
\end{example}

\section{The Converse is Not True \label{secBeyond}}

Note that information-distributive networks don't subsume all possible routing-optimal networks.
In the following, we show an example of such a network.

\begin{example}
Consider the network as shown in Fig. \ref{figBeyond}.
We first show that it is not information-distributive.
Define the following paths:
\begin{flalign*}
& P_{11} = \{a_1,e_1,b_1\}, P_{12} = \{a_2, e_3, b_2\} \\
& P_{21} = \{a_3,e_3,b_3\}, P_{22} = \{a_4, e_5, b_4\} \\
& P_{31} = \{a_5,e_5,b_5\}, P_{32} = \{a_6, e_6, b_6\}
\end{flalign*}
For $1\le i \le 3$, let $\C{P}_i=\{P_{i1},P_{i2}\}$, and $\C{K}=(\C{P}_i)^3_{i=1}$.
Since each source has only two outgoing edges, $\C{K}$ is the only-possible path-set sequence.
It can be verified that for all cumulative and distributive cut-set sequences, $\C{K}$ is not extendable.
For instance, let $C_1=\{a_1,e_3\}$, $C_2=\{e_3,b_4\}$, and $C_3=\{e_5,b_6\}$.
Clearly, the cut-set sequence $\C{W}=(C_i)^3_{i=1}$ is cumulative and distributive.
However, it can be seen that $P_{22}$ overlaps with $P_{31}$, but $P_{22} \cap C_2 = \{b_4\}$, and $P_{31} \cap C_3 = \{e_5\}$.
Hence, $\C{K}$ doesn't satisfy the condition of Definition \ref{defExtendable}.
Similarly, we can verify other cases.
Thus, the network is not information-distributive.

\begin{figure}
\centering
\includegraphics[scale=.6]{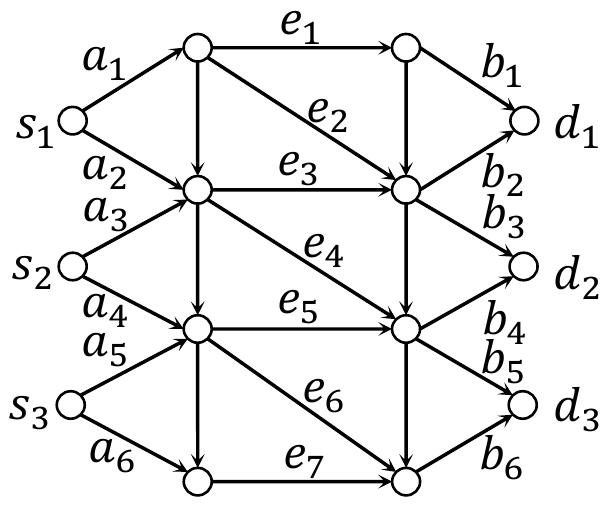}
\caption{A routing-optimal network that is not information-distributive. \label{figBeyond}}
\end{figure}

Nevertheless, we can show that the network is routing-optimal.
Consider an arbitrary rate vector $\B{R}=(R'_1,R'_2,R'_3)\in \C{R}_{nc}$.
For $\epsilon=\frac{1}{k}$ ($k\ge 2$), there exists a network code of length $n$ such that (\ref{eqNCAchieveI1})-(\ref{eqNCAchieveI3}) are satisfied.

Define the following cut-sets, and permutations of edges:
\begin{flalign*}
& C_1 = \{e_1,e_2,e_3\}\; C_2=\{e_3,e_4,e_5\}\; C_3=\{e_5,e_6,e_7\} \\
& T_1 = (e_1,e_2,e_3)\hspace{6pt} T_2=(e_3,e_4,e_5)\hspace{6pt} T_3=(e_5,e_6,e_7)
\end{flalign*}
Define the following permutations:
\begin{flalign*}
T'_1 = (b_1,b_2) \; T'_2 = (b_3,b_4) \; T'_3 = (b_5,b_6)
\end{flalign*}
Let $\C{W}=(C_i)^3_{i=1}$, and $\C{T}=(T_i)^3_{i=1}$.
Clearly, $\C{W}$ satisfies the condition of Definition \ref{defCumulative}.
Thus, according to Lemma \ref{lemmaDetermine}, for $i=1,2,3$, the following inequality holds:
\begin{flalign}
\label{eqEx41}
I(Y_i; U_{\inedge(d_i)} | Y_{1:i-1}) \le I(U_{C_i}; Y_i | Y_{1:i-1})
\end{flalign}
Moreover, since $\C{T}$ satisfies the conditions of Definition \ref{defDistributive}.
By Lemma \ref{lemmaInfoDistr}, for $e\in \bigcup^3_{i=1}C_i$, the following inequality holds:
\begin{flalign}
\label{eqEx42}
\sum^3_{i=1} \sum_{e\in C_i} I(Y_i; U_e | Y_{1:i-1}, U_{T_i(e)}) \le H(U_e)
\end{flalign}

Define the following paths:
\begin{flalign*}
& P_{11} = \{a_1,e_1,b_1\}, P_{12} = \{a_1,e_2,b_2\}, P_{13} = \{a_2,e_3,b_2\} \\
& P_{21} = \{a_3,e_3,b_3\}, P_{22} = \{a_3,e_4,b_4\}, P_{23} = \{a_4,e_5,b_4\} \\
& P_{31} = \{a_5,e_5,b_5\}, P_{32} = \{a_5,e_6,b_6\}, P_{33} = \{a_6,e_7,b_6\}
\end{flalign*}
Let $\C{P}_i=\{P_{i1},P_{i2},P_{i3}\}$.
Define a routing scheme as follows:
\begin{flalign*}
f^{n,k}_i(P) = \begin{cases}
\frac{1}{n} I(Y_i; U_{P\cap C_i} | Y_{1:i-1}, U_{T_i(e)}) & \text{if } P\in \C{P}_i \\
0 & \text{otherwise.}
\end{cases}
\end{flalign*}

Note the following inequalities hold for $i=1,2,3$:
\begin{flalign}
\begin{split}
\label{eqEx44}
& \frac{1}{n}H(Y_i) \ge \sum_{P\in \C{P}_i} f^{n,k}_i(P) \\
=& \frac{1}{n} I(Y_i; U_{C_i} | Y_{1:i-1}) \overset{(a)}{\ge}  \frac{1}{n} I(Y_i; U_{\inedge(d_i)} | Y_{1:i-1}) \\
=& \frac{1}{n}\sum_{e\in \inedge(d_i)} I(Y_i; U_e | Y_{1:i-1}, U_{T'_i(e)}) \\ \overset{(b)}{\ge}& \frac{1}{n}(1-\frac{1}{k})H(Y_i) - \frac{1}{n} \ge (1-\frac{1}{k})(R'_2-\frac{1}{k}) - \frac{1}{n}
\end{split}
\end{flalign}
where $(a)$ holds because $U_{\inedge(d_i)}$ is a function of $U_{C_2},Y_{1:i-1}$;
$(b)$ is due to Fano's Inequality.
For $i=1,2,3$, $e'\in \bigcup^3_{i=1}C_i$, and $e\in \inedge(d_i)$, define the following notations:
\begin{flalign*}
&y^{n,k}_i = \frac{1}{n}H(Y_i) \quad u^{n,k}_{e'} = \frac{1}{n}H(U_{e'}) \\
&g^{n,k}_{i,e} = \frac{1}{n}I(Y_i; U_e | Y_{1:i-1}, U_{T'_i(e)})
\end{flalign*}
Thus, (\ref{eqEx44}) can be rewritten in a concise form as:
\begin{flalign}
\begin{split}
\label{eqEx45}
& y^{n,k}_i \ge \sum_{P\in \C{P}_i} f^{n,k}_i(P) \ge \sum_{e\in \inedge(d_i)}g^{n,k}_{i,e} \\
\ge& (1-\frac{1}{k})y^{n,k}_i - \frac{1}{n} \ge (1-\frac{1}{k})(R'_2-\frac{1}{k}) - \frac{1}{n}
\end{split}
\end{flalign}
Due to (\ref{eqEx44}), it can be seen that:
\begin{flalign*}
& \frac{1}{2}y^{n,k}_i - 1 \le (1-\frac{1}{k})y^{n,k}_i - \frac{1}{n} \\
\le & \frac{1}{n} I(Y_i; U_{\inedge(d_i)} | Y_{1:i-1}) \\
\le & \frac{1}{n} \sum_{e\in \inedge(d_i)} H(U_e) \le 2(1+\frac{1}{k}) \le 3
\end{flalign*}
This means that $y^{n,k}_i \le 8$.
Clearly, all $y^{n,k}_i$'s, $u^{n,k}_{e'}$'s, $g^{n,k}_{i,e}$'s and $f^{n,k}_i(P)$'s have upper bounds.
Thus, there exists a sub-sequence $(n_l,k_l)^{\infty}_{l=1}$ such that $y^{n_l,k_l}_i$, $u^{n_l,k_l}_{e'}$, $g^{n_l,k_l}_{i,e}$ and $f^{n_l,k_l}_i(P)$ approach finite limits when $l\rightarrow \infty$.
Define the following notations:
\begin{flalign*}
y_i = \lim_{l\rightarrow \infty} y^{n_l,k_l}_i \quad u_{e'} = \lim_{l\rightarrow \infty} u^{n_l,k_l}_{e'} \quad g_{i,e} = \lim_{l\rightarrow \infty} g^{n_l,k_l}_{i,e}
\end{flalign*}
Clearly, the following inequalities holds:
\begin{flalign*}
u_{e'} \le 1 \quad g_{i,e} \le u_e \le 1
\end{flalign*}
Define the following routing scheme:
\begin{flalign*}
f_i(P) = \begin{cases}
\lim_{l\rightarrow \infty} f^{n_l,k_l}_i(P) & \text{if } P\in \C{P}_i \\
0 & \text{otherwise}.
\end{cases}
\end{flalign*}
We will prove that this routing scheme satisfies (\ref{eqRoutingCond1}) and (\ref{eqRoutingCond2}).
According to (\ref{eqEx44}), we see that $\sum_{P\in \C{P}_i} f_i(P) \ge R'_2$, and thus, (\ref{eqRoutingCond1}) is satisfied.
Moreover, due to (\ref{eqEx42}), (\ref{eqRoutingCond2}) is satisfied for $e\in \bigcup^3_{i=1}C_i$.
For $a_3$, we have:
\begin{flalign*}
& f^{n,k}_2(P_{21}) + f^{n,k}_2(P_{22}) \\
=& \frac{1}{n} I(Y_2; U_{\{e_3,e_4\}} | Y_1) \overset{(c)}{\le} \frac{1}{n} I(Y_2; U_{a_3} | Y_1) \\
\le& H(U_{a_3}) \le 1 + \frac{1}{k}
\end{flalign*}
where $(c)$ is due to the fact that $U_{\{e_3,e_4\}}$ is a function of $U_{a_3},Y_1$.
Thus, $f_2(P_{21}) + f_2(P_{22}) \le 1$, and (\ref{eqRoutingCond2}) is satisfied for $a_3$.
Using similar arguments, we can prove that (\ref{eqRoutingCond2}) is satisfied for $a_1,a_5$.
Now consider $b_4$.
Due to (\ref{eqEx45}), the following equations hold:
\begin{flalign}
\label{eqEx46}
y_2 = g_{2,b_3} + g_{2,b_4} = f_2(P_{21}) + f_2(P_{22}) + f_2(P_{23})
\end{flalign}
Meanwhile, since $U_{b_3}$ is a function of $U_{e_3}, Y_1$, the following equations hold:
\begin{flalign*}
f^{n,k}_2(P_{21}) = \frac{1}{n} I(Y_2;U_{e_3}|Y_1) \ge \frac{1}{n} I(Y_2;U_{b_3}|Y_1) = g^{n,k}_{2,b_3}
\end{flalign*}
Hence, $f_2(P_{21}) \ge g_{2,b_3}$.
Combining with (\ref{eqEx46}), we have:
\begin{flalign*}
f_2(P_{22}) + f_2(P_{23}) \le g_{2,b_3} \le 1
\end{flalign*}
Hence, (\ref{eqRoutingCond2}) holds for $b_3$.
Similarly, we can prove that (\ref{eqRoutingCond2}) holds for $b_2,b_6$.
It can be easily seen that for all the other edges, (\ref{eqRoutingCond2}) also holds.
Therefore, we have proved that $\B{R}\in \C{R}_r$.
This means that $\C{R}_{nc} \subseteq \C{R}_r$, and the network is routing-optimal.
\myqed
\end{example}

\section{Conclusion \label{secConclusion}}

In this paper, we present a class of routing-optimal networks, called information-distributive networks, defined by three topological features.
Due to these features, there is always a routing scheme that achieves the same rate vector as network coding such that the traffic transmitted through the network is the information distributed over the cut-sets between the sources and the sinks in the corresponding network coding scheme.
We then present some examples of information-distributive networks related to index coding and single unicast with hard deadline constraint.

\section*{Acknowledgment}
The authors would like to thank Minghua Chen for insightful discussions on the problem of network coding with hard deadline constraints, which inspired this follow-up work (see Subsection IV.B). Chun Meng was visiting, and was supported by, the Network Coding Institute of Hong Kong at that point (Aug. 2012 - June 2013).

This work was supported by NSF Awards 0747110 (CAREER) and 1028394, AFOSR MURI Award FA9550-09-0643.

\appendices

\section{Useful Tools \label{appTools}}

In this section, we present some useful tools to be used in the sequel.

\begin{proposition}
\label{propIt1}
The following equations hold:
\begin{enumerate}
\item $H(X|Y) = H(X|Y,f(Y))$.
\item $I(X;Y|Z) = I(X;Y|Z,f(Z))$.
\item $H(X|f(Y)) \ge H(X|Y)$.
\item $I(X;Y|Z,W) \ge I(X;f(Y,Z)|Z,W)$.
\end{enumerate}
\end{proposition}
\begin{proof}
1) The following equation holds:
\begin{flalign}
\label{eqPropIt11}
\begin{split}
H(X,Y,f(Y)) =& H(Y) + H(X|Y) + H(f(Y)|X,Y) \\
=& H(Y) + H(X|Y)
\end{split}
\end{flalign}
Meanwhile, we have:
\begin{flalign}
\label{eqPropIt12}
\begin{split}
H(X,Y,f(Y)) =& H(Y) + H(f(Y)|Y) + H(X|Y,f(Y)) \\
=& H(Y) + H(X|Y,f(Y))
\end{split}
\end{flalign}
Combining Eq. (\ref{eqPropIt11}) and Eq. (\ref{eqPropIt12}), we have $H(X|Y) = H(X|Y,f(Y))$.

2) Due to 1), we can derive:
\begin{flalign*}
I(X;Y|Z,f(Z)) =& H(X|Z,f(Z)) - H(X|Y,Z,f(Z)) \\
=& H(X|Z) - H(X|Y,Z) \\
=& I(X;Y|Z)
\end{flalign*}

3) First, the following equalities hold:
\begin{flalign}
\label{eqPropIt13}
H(X,Y,f(Y)) = H(f(Y)) + H(X|f(Y)) + H(Y|X,f(Y))
\end{flalign}
Combining Eq. (\ref{eqPropIt11}) and Eq. (\ref{eqPropIt13}), we then have:
\begin{flalign*}
\begin{split}
H(X|f(Y)) =& H(X|Y) + H(Y) - H(f(Y)) - H(Y|X,f(Y)) \\
\overset{(a)}{=}& H(X|Y) + H(Y|f(Y)) - H(Y|X,f(Y)) \\
=& H(X|Y) + I(X;Y|f(Y)) \\
\ge & H(X|Y)
\end{split}
\end{flalign*}
where $(a)$ follows from the equation: $H(Y)=H(Y,f(Y))=H(f(Y))+H(Y|f(Y))$.

4) We have the following equations:
\begin{flalign*}
& I(X;Y|Z,W) - I(X;f(Y,Z)|Z,W) \\
=& H(X|Z,W) - H(X|Y,Z,W) - \\
& \quad [H(X|Z,W) - H(X|f(Y,Z),Z,W)] \\
=& H(X|f(Y,Z),Z,W) - H(X|Y,Z,W) \ge 0
\end{flalign*}
where the last inequality is due to 3) and the fact that $(f(Y,Z),Z,W)$ is a function of $(Y,Z,W)$.
\end{proof}

\begin{proposition}
\label{propMutualInfoIneq3}
If $Y\rightarrow (X,W) \rightarrow Z$, then $I(X;Y|W) \ge I(X;Y|W,Z)$ and $I(X;Y|W) \ge I(Z;Y|W)$.
As a special case, we have $I(X;Y|W) \ge I(X;Y|W,f(X,W))$ and $I(X;Y|W) \ge I(f(X,W);Y|W)$.
\end{proposition}
\begin{proof}
We have the following equations:
\begin{flalign*}
& I(X,Z;Y|W) = I(Z;Y|W) + I(X;Y|W,Z) \\
=& I(X;Y|W) + I(Z;Y|X,W) = I(X;Y|W)
\end{flalign*}
Thus, it must be that $I(X;Y|W) \ge I(X;Y|W,Z)$ and $I(X;Y|W) \ge I(Z;Y|W)$.
Since the following chain: $Y\rightarrow (X,W) \rightarrow f(X,W)$ holds, 
we must have $I(X;Y|W) \ge I(X;Y|W,f(X,W))$ and $I(X;Y|W) \ge I(f(X,W);Y|W)$.
\end{proof}

\section{Proofs for Information-Distributive Networks \label{appInfoDistr}}

\begin{proof}[Proof of Lemma \ref{lemmaDetermine}]
Let $S'_i$ denote the set consisting of the outgoing edges of $s_1,\cdots,s_{i-1}$.
Since each path from $s_j$ ($i \le j < K$) to $d_i$ must pass through an edge in $C_i$, $S'_i\cup C_i$ forms a cut-set between $\{s_1,\cdots,s_K\}$ and $d_i$.
Thus $U_{\inedge(d_i)}$ is a function of $U_{S'_i}, U_{C_i}$.
Meanwhile, $Y_{S'_i}$ is a function of $Y_{1:i-1}$.
Thus, $U_{\inedge(d_i)}$ is a function of $Y_{1:i-1}, U_{C_i}$.
According to Proposition \ref{propMutualInfoIneq3}, (\ref{eqDetermine}) holds.
\end{proof}

\begin{proof}[Proof of Lemma \ref{lemmaInfoDistr}]
Let $\C{T}$ be the permutation sequence as defined in Definition \ref{defDistributive}.
Consider an arbitrary edge $e\in \bigcup^K_{i=1} C_i$.
Without loss of generality, let $\C{W}(e) = \{C_{n_1},\cdots,C_{n_k}\}$, where $1\le n_1 < \cdots < n_k \le K$.
Then we have:
\begin{flalign*}
& \sum_{1\le i \le K, e\in C_i} I(Y_i;U_e | Y_{1:i-1},U_{T_i(e)}) \\
=& \sum^k_{i=1} I(Y_i;U_e | Y_{1:n_i-1}, U_{T_{n_i}(e)})
\end{flalign*}
For $k=1$, the following equation holds:
\begin{flalign*}
&\sum^k_{i=1}  I(Y_i;U_e | Y_{1:n_i-1}, U_{T_{n_i}(e)}) \\
=& I(Y_i:U_e | Y_{1:n_1-1}, U_{T_{n_1}(e)}) \le H(U_e)
\end{flalign*}
Hence, (\ref{eqTotalDistr}) holds for $k=1$.
We now consider the case $k>1$.
We will prove the following inequality holds for $1\le p \le k$:
\begin{flalign}
\begin{split}
\label{eq_inequality_1}
&\sum^k_{i=p} I(Y_i; U_e | Y_{1:n_i-1}, U_{T_{n_i}(e)}) \\
\le& I(Y_{n_p:n_k}; U_e | Y_{1:n_p-1}, U_{T_{n_p}(e)})
\end{split}
\end{flalign}
Clearly, (\ref{eq_inequality_1}) holds trivially for $p=k$.
Assume it holds for $p>1$.
We will prove it also holds for $p-1$.
\begin{flalign*}
& \sum^k_{i=p-1} I(Y_i; U_e | Y_{1:n_i-1}, U_{T_{n_i}(e)}) \\
\overset{(a)}{\le} & I(Y_{n_p:n_k}; U_e | Y_{1:n_p-1}, U_{T_{n_p}(e)}) + \\
& \hspace{.1cm} I(Y_{n_{p-1}}; U_e | Y_{1:n_{p-1}-1}, U_{T_{n_{p-1}}(e)}) \\
\overset{(b)}{=} & I(Y_{n_p:n_k}; U_e | Y_{1:n_p-1}, U_{T_{n_p}(e) \cap T_{n_{p-1}}(e)}, U_{T_{n_p}(e) - T_{n_{p-1}}(e)}) \\
& \hspace{.1cm}  + I(Y_{n_{p-1}}; U_e | Y_{1:n_{p-1}-1}, U_{T_{n_p}(e) \cap T_{n_{p-1}}(e)}, U_{T_{n_{p-1}}(e) - T_{n_p}(e)}) \\
\overset{(c)}{\le} & I(Y_{n_p:n_k}; U_e | Y_{1:n_p-1}, U_{T_{n_p}(e) \cap T_{n_{p-1}}(e)}) + \\
& \hspace{.1cm} I(Y_{n_{p-1}}; U_e | Y_{1:n_{p-1}-1}, U_{T_{n_p}(e) \cap T_{n_{p-1}}(e)}, U_{T_{n_{p-1}}(e) - T_{n_p}(e)}) \\
\overset{(d)}{=} & I(Y_{n_p:n_k}; U_e | Y_{1:n_p-1}, U_{T_{n_p}(e) \cap T_{n_{p-1}}(e)}, U_{T_{n_{p-1}}(e) - T_{n_p}(e)}) \\
& \hspace{.1cm} + I(Y_{n_{p-1}}; U_e | Y_{1:n_{p-1}-1}, U_{T_{n_p}(e) \cap T_{n_{p-1}}(e)}, U_{T_{n_{p-1}}(e) - T_{n_p}(e)}) \\
\overset{(e)}{=} & I(Y_{n_p:n_k}; U_e | Y_{1:n_p-1}, U_{T_{n_{p-1}}(e)}) + \\
& \hspace{.1cm} I(Y_{n_{p-1}}; U_e | Y_{1:n_{p-1}-1}, U_{T_{n_{p-1}}(e)}) \\
\le & I(Y_{n_p:n_k}; U_e | Y_{1:n_p-1}, U_{T_{n_{p-1}}(e)}) + \\
& \hspace{.1cm} I(Y_{n_{p-1}:n_p-1}; U_e | Y_{1:n_{p-1}-1}, U_{T_{n_{p-1}}(e)}) \\
\overset{(f)}{=} & I(Y_{n_{p-1}:n_k}; U_e | Y_{1:n_{p-1}-1}, U_{T_{n_{p-1}}(e)})
\end{flalign*}
where $(a)$ is due to our assumption that (\ref{eq_inequality_1}) holds for $p$;
$(b)$ is due to the equalities, $T_{n_p}(e) = (T_{n_p}(e) \cap T_{n_{p-1}}(e)) \cup (T_{n_p}(e) - T_{n_{p-1}}(e))$ and $T_{n_{p-1}}(e) = (T_{n_p}(e) \cap T_{n_{p-1}}(e)) \cup (T_{n_{p-1}}(e) - T_{n_p}(e))$;
$(c)$ is due to our premise that $\C{W}$ is distributive: for each $e'\in T_{n_p}(e) - T_{n_{p-1}}(e)$, $\alpha(e') \le n_k$, and thus $U_{e'}$ is a function of $Y_{1:n_k}$; therefore, according to Proposition \ref{propMutualInfoIneq3}, we have:
\begin{flalign*}
&I(Y_{n_p:n_k}; U_e | Y_{1:n_p-1}, U_{T_{n_p}(e) \cap T_{n_{p-1}}(e)}, U_{T_{n_p}(e) - T_{n_{p-1}}(e)}) \\
&\le I(Y_{n_p:n_k}; U_e | Y_{1:n_p-1}, U_{T_{n_p}(e) - T_{n_{p-1}}(e)})
\end{flalign*}
$(d)$ is also due to our premise that $\C{W}$ is distributive: for each $e'\in T_{n_{p-1}}(e) - T_{n_p}(e)$, $\alpha(e') \le n_p - 1$, and thus $U_{e'}$ is  a function of $Y_{1:n_p-1}$; therefore, the following equality holds according to Proposition \ref{propIt1}:
\begin{flalign*}
&I(Y_{n_{p-1}}; U_e | Y_{1:n_{p-1}-1}, U_{T_{n_p}(e) \cap T_{n_{p-1}}(e)}, U_{T_{n_{p-1}}(e) - T_{n_p}(e)}) \\
&= I(Y_{n_{p-1}}; U_e | Y_{1:n_{p-1}-1}, U_{T_{n_p}(e) \cap T_{n_{p-1}}(e)})
\end{flalign*}
$(e)$ is again due to $T_{n_{p-1}}(e) = (T_{n_p}(e) \cap T_{n_{p-1}}(e)) \cup (T_{n_{p-1}}(e) - T_{n_p}(e))$;
$(f)$ is due to chain rule of mutual information.
Thus, (\ref{eq_inequality_1}) holds for $p-1$.
This means that (\ref{eq_inequality_1}) must hold for all $1\le p \le k$.
Letting $p=1$ in (\ref{eq_inequality_1}), we have:
\begin{flalign*}
& \sum^k_{i=1} I(Y_i; U_e | Y_{1:n_i-1}, U_{T_{n_i}(e)}) \\
\le & I(Y_{n_1:n_k}; U_e | Y_{1:n_1-1}, U_{T_{n_1}(e)}) \le H(U_e)
\end{flalign*}
Thus, the lemma holds.
\end{proof}

Let $e$ be an edge that is passed through by at least one path in an extendable path-set sequence $\C{K}$.
According to the above definition, all the paths in $\C{K}$ that pass through $e$ must pass through a single edge in $\C{W}$.
We use $\mu_e$ to denote this edge, and refer to it as the \textit{representative} of $e$ in $\C{W}$.

\begin{proof}[Proof of Theorem \ref{thInfodistr}]
Let $\C{W}=\{C_i:1\le i \le K\}$ be a cumulative and distributive cut-set sequence, $\C{T}$ a permutation sequence for $\C{W}$ that satisfies the conditions of Definition \ref{defDistributive}, and $\C{K}=\{\C{P}_i: 1\le i \le K\}$ an extendable path-set sequence for $\C{W}$.
Let $\B{R}=(R'_i: 1\le i \le K)$ be an arbitrary rate vector in $\C{R}_{nc}$.
Therefore, for $\epsilon = \frac{1}{k} > 0$ ($k\in \BB{Z}_{> 0}$), there exists a network code which satisfies (\ref{eqNCAchieve1})-(\ref{eqNCAchieve3}).
In the rest of this proof, all the random variables are defined in this network code.

We then define the following routing scheme: for $1\le i \le K$,
\begin{flalign*}
f^{n,k}_i(P) = \begin{cases}
\frac{1}{n} I(Y_i; U_e | Y_{1:i-1}, U_{T_i(e)}) & \text{if } P\in \C{P}_i, e\in P\cap C_i; \\
0 & \text{otherwise}.
\end{cases}
\end{flalign*}
Since $\C{W}$ is cumulative, the following equation holds:
\begin{flalign}
\label{eqProofInfoDistr1}
\begin{split}
& \sum_{P\in \C{P}_{s_id_i}} f^{n,k}_i(P) = \sum_{P\in \C{P}_i} f^{n,k}_i(P) \\
=& \frac{1}{n} \sum_{e\in C_i} I(Y_i; U_e | Y_{1:i-1}, U_{T_i(e)}) \\
\overset{(a)}{=} & \frac{1}{n} I(Y_i; U_{C_i} | Y_{1:i-1}) \overset{(b)}{\ge} \frac{1}{n} I(Y_i; U_{\inedge(d_i)} | Y_{1:i-1})
\end{split}
\end{flalign}
where $(a)$ is due to (\ref{eqInfoDistr}), and $(b)$ is due to (\ref{eqDetermine}).
Define $\delta'_i = Pr(Y_i\text{ cannot be decoded from } U_{\inedge(d_i)}, Y_{1:i-1})$.
Clearly, $\delta'_i \le \delta_i \le \frac{1}{k}$.
Then, we can derive the following equation:
\begin{flalign*}
& \frac{1}{n} I(Y_i;U_{\inedge(d_i)}|Y_{1:i-1}) \\
=& \frac{1}{n}(H(Y_i|Y_{1:i-1}) - H(Y_i|U_{\inedge(d_i)}, Y_{1:i-1})) \\
\overset{(c)}{=}& \frac{1}{n}(H(Y_i) - H(Y_i|U_{\inedge(d_i)}, Y_{1:i-1})) \\
\overset{(d)}{\ge}& \frac{1}{n}(H(Y_i) - 1 - \delta'_i\log|\C{Y}_i|) \\
=&   (1-\delta'_i)\frac{1}{n} H(Y_i) - \frac{1}{n} \\
\overset{(e)}{\ge}& \bigg(1-\frac{1}{k}\bigg)\bigg(R'_i-\frac{1}{k}\bigg) - \frac{1}{n}
\end{flalign*}
where $(c)$ is due to the fact that $Y_i$ is independent from $Y_{1:i-1}$;
$(d)$ is due to Fano Inequality;
$(e)$ is due to (\ref{eqNCAchieveI2}).
Combining the above equation with (\ref{eqProofInfoDistr1}), the following inequality holds:
\begin{flalign}
\label{eqProofInfoDistr2}
\sum_{P\in \C{P}_{s_id_i}} f^{n,k}_i(P) \ge \bigg(1-\frac{1}{k}\bigg)\bigg(R'_i-\frac{1}{k}\bigg) - \frac{1}{n}
\end{flalign}
Let $e$ be an edge that is passed through by at least one path in $\C{K}$.
Since $\C{K}$ is extendable, the paths in $\C{K}$ that pass through $e$ must pass through $e$'s representative $\mu_e$ in $\C{W}$.
Hence, the following equation holds:
\begin{flalign}
\label{eqProofInfoDistr3}
\begin{split}
& \sum^K_{i=1}\sum_{P\in \C{P}_{s_id_i},e\in P} f^{n,k}_i(P) \\
=& \sum^K_{i=1}\sum_{P\in \C{P}_i,e\in P} f^{n,k}_i(P) \\
\le& \sum^K_{i=1}\sum_{P\in \C{P}_i, \mu_e\in P} f^{n,k}_i(P) \\
=& \frac{1}{n} \sum_{1\le i \le K, \mu_e\in C_i} I(Y_i;U_{\mu_e} | Y_{1:i-1}, U_{T_i(\mu_e)}) \\
\overset{(f)}{\le}& \frac{1}{n} H(U_{\mu_e}) \overset{(g)}{\le} 1 + \frac{1}{k}
\end{split}
\end{flalign}
where $(f)$ is due to (\ref{eqTotalDistr}); $(g)$ is due to (\ref{eqNCAchieveI1}).

Since each $f^{n,k}_i(P)$ has an upper bound, there exists a sequence $(n_l,k_l)^{\infty}_{l=1}$ such that for $1\le i \le K$, the sequence $(f^{n_l,k_l}_i(P))^{\infty}_{l=1}$ approaches a finite limit.
Define the following routing scheme:
\begin{flalign*}
f_i(P) = \begin{cases}
\lim_{l\rightarrow \infty} f^{n_l,k_l}_i(P) & \text{if } P \in \C{P}_i \\
0 & \text{otherwise}.
\end{cases}
\end{flalign*}
Due to (\ref{eqProofInfoDistr2}) and (\ref{eqProofInfoDistr3}), $f_i(P)$ satisfies (\ref{eqRoutingCond1}) and (\ref{eqRoutingCond2}).
Hence, $\B{R} \in \C{R}_r$.
This implies that $\C{R}_{nc} \subseteq \C{R}_r$, and the network is routing-optimal.
\end{proof}

\section{Proofs for Examples \label{appProofEx}}

\begin{proof}[Proof of Theorem \ref{thIndexCode}]
Assume $\C{W}$ is cumulative.
Hence, $G_1$ is information-distributive
According to Theorem \ref{thInfodistr}, $G_1$ is routing-optimal.
Since routing can achieve a common rate of at most $\frac{1}{K}$, $l_{min} =mK$.

Now assume $l_{min}=mK$.
We consider a side-information graph $G'=(V',E')$ \cite{index_code}, where $V'=\{1,\cdots,K\}$, and $E'=\{(j,i): X_i\in \C{H}_j,1\le i,j\le K\}$.
It has been shown that if $l_{min}=mK$, then $G'$ is acyclic \cite{index_code}.
We will show that $\C{W}$ is information-distributive.
Since $G'$ is acyclic, we can re-index the nodes in $G'$, such that if $(j,i)\in E'$, $j<i$.
Let $1\le i<j\le K$.
Consider a path $P$ from $s_j$ to $d_i$.
Since $(j,i)\notin E'$, $X_j\notin \C{H}_i$.
Thus, there is no directed edge from $s_j$ to $d_i$ in $G_1$, and $P$ must pass through $(u,v)\in C_i$.
Hence, $\C{W}$ is cumulative, and $G_1$ is information-distributive.
\end{proof}

\begin{proof}[Proof of Lemma \ref{lemmaDelayCumulative}]
Let $0\le i<j\le K$.
Assume there is a directed path $P$ from $s_j$ to $d_i$.
Let $P_1$ be the part of $P$ after $s[j]$.
Clearly, $P'=\{(s_i,s[i]),(s[i],s[i+1]),\cdots,(s[j-1],s[j])\} \cup P_1$ is a directed path from $s_i$ to $d_i$.
Since $C[i]$ is a cut-set between $s[i]$ and $d[i]$, $P'$ must pass through an edge $e[k]\in C[i]$.
Thus, $e[k] \in P$.
This means that $\C{W}$ is cumulative.
\end{proof}

Since the duration between $e[t]$ and $s[\alpha(e[t])]$ is $\delta(e)$, we have: 
\begin{flalign}
\label{eqDelayAlpha}
\alpha(e[t]) = t - \delta(e)
\end{flalign}

\begin{lemma}
\label{lemmaDelayDistributive}
If $C[0]$ is distributive, $\C{W}$ is distributive.
\end{lemma}
\begin{proof}
Let $T[t]=(e_i[t_i+t])^k_{i=1}$, and define a permutation sequence $\C{T}=(T[t])^K_{t=0}$ for $\C{W}$.
We will prove that if $C[0]$ is distributive, $\C{T}$ satisfies (\ref{eqDistr1}) and (\ref{eqDistr2}).

Consider an edge $e_p[t_p]\in C[0]$.
Let$(e_p[t_{n_i}])^k_{i=1}$ be the recurrent sequence in $C[0]$, in which all the edges are time-shifted versions of $e_p$.
Without loss of generality, let $n_j=p$.
Next, consider $e_p[t_p+k]\in C[k]$.
Let $\C{W}(e_p[t_p+k])=\{C[t]:e_p[t_p+k] \in C[t], 0\le t\le K\}$ denote the subset of cut-sets which contain $e_p[t_p+k]$.
Clearly, $C[k-t_{n_{j+1}} + t_{n_j}]$ and $C[k+t_{n_j}-t_{n_{j-1}}]$ are the cut-sets in $\C{W}(e_p[t_p+k])$ that lies immediately before and after $C[k]$ respectively, and $C[k+t_{n_j}-t_{n_1}]$ is the last cut-set in $\C{W}(e_p[t_p+k])$.

Consider an edge $e_q[t_q+k]\in C[k]$ be an edge that lies before $e_p[t_p+k]$ in $T[k]$, but doesn't appear before $e_p[t_p+k]$ in $T[k-t_{n_{j+1}}+t_{n_j}]$.
This means that $e_q[t_q+t_{n_{j+1}}-t_{n_j}] \notin C[0]$.
Thus, the following equation holds:
\begin{flalign*}
\alpha(e_q[t_q+k]) = k + t_q - \delta(e_q) \overset{(a)}{\le} k + t_{n_j} - t_{n_1}.
\end{flalign*}
where $(a)$ is due to the premise that $C[0]$ is distributive.
Hence, (\ref{eqDistr1}) is satisfied.

Now assume that $e_q[t_q+k]\in C[k]$ lies before $e_p[t_p+k]$ in $T[k]$, but doesn't appear before $e_p[t_p+k]$ in $T[k-t_{n_j}+t_{n_{j-1}}]$.
This implies that $e_q[t_q-t_{n_j} + t_{n_{j-1}}] \notin C[0]$.
Thus, the following equation holds:
\begin{flalign*}
\alpha(e_q[t_q+k]) = k + t_q - \delta(e_q) \overset{(b)}{\le} k + t_{n_j} - t_{n_{j-1}} - 1
\end{flalign*}
where $(b)$ is again due to the premise that $C[0]$ is distributive.
Hence, (\ref{eqDistr2}) is satisfied.
$\C{W}$ is distributive.
\end{proof}

\begin{lemma}
\label{lemmaDelayExtendable}
If $\C{P}$ is extendable, $\C{K}$ is extendable.
\end{lemma}
\begin{proof}
Consider two paths $P_i,P_j\in \C{P}$.
Assume $P_i[k_1]$ overlaps with $P_j[k_2]$ at $e[t]$.
Thus, $e[t-k_1]\in P_i$ and $e[t-k_2]\in P_j$.
Since $\C{P}$ is extendable, this means that $e_i=e_j$ and
\begin{flalign*}
t_i-t_j = t-k_1 - (t-k_2) = k_2-k_1
\end{flalign*}
Note that $e_i[t_i+k_1]$ is the edge in $\C{W}$ that is passed through by $P_i[k_1]$.
We have:
\begin{flalign*}
e_i[t_i+k_1] = e_j[t_j+k_2] \in P_j[k_2] \cap C[k_2].
\end{flalign*}
Thus, $P_i[k_1]$ and $P_j[k_2]$ pass through the same edge $e_i[t_i+k_1]$ in $\C{W}$.
Hence, $\C{K}$ is extendable.
\end{proof}

\begin{proof}[Proof of Theorem \ref{thDelayInfoDistributive}]
Due to Lemmas \ref{lemmaDelayExtendable}, \ref{lemmaDelayCumulative} and \ref{lemmaDelayDistributive}, the theorem holds.
\end{proof}

\bibliographystyle{IEEEtran}


\end{document}